\documentclass[11pt]{article}

\usepackage[margin = 2.54cm]{geometry}

\usepackage{microtype}
\usepackage{epsfig}
\usepackage{microtype}
\usepackage{graphicx}
\usepackage{enumitem}
\usepackage{epsfig,enumerate,amsmath,amsfonts,amssymb,amsthm,mathrsfs,ifpdf}
\usepackage{indentfirst,relsize}
\usepackage[numbers]{natbib}
\usepackage[linesnumbered,vlined, ruled]{algorithm2e}
\usepackage{setspace}
\usepackage{enumerate}
\usepackage{latexsym}
\usepackage{stackrel}
\usepackage[all]{xy}
\usepackage[usenames,dvipsnames]{pstricks}
\usepackage{pst-grad} 
\usepackage{pst-plot} 
\usepackage{xspace}

\newcommand{\size}[1]{\left| #1 \right|}

\newcommand{\E}{\mathbb{E}}
\renewcommand{\hat}{\widehat}
\newcommand{\remove}[1]{}

\newcommand{\R}{\mathbb{R}}
\newcommand{\N}{\mathbb{N}}

\newcommand{\cS}{\mathcal{S}}
\newcommand{\cT}{\mathcal{T}}

\newcommand{\cA}{\mathcal{A}}

\newcommand{\cC}{\mathcal{C}}

\newcommand{\cE}{\mathcal{E}}
\newcommand{\cF}{\mathcal{F}}
\newcommand{\cP}{\mathcal{P}}

\newcommand{\cG}{\mathcal{G}}
\newcommand{\cR}{\mathcal{R}}
\newcommand{\Oh}{\mathcal{O}}

\newcommand{\tOh}{\widetilde{{\mathcal O}}}

\newcommand{\cH}{\mathcal{H}}

\newcommand{\hs}{HS}

\newcommand{\eps}{\epsilon}
\newcommand{\pr}{\mathbb{P}}

\theoremstyle{plain}
\newtheorem{theo}{Theorem}[section]
\newtheorem{lem}[theo]{Lemma}
\newtheorem{pre}[theo]{Proposition}
\newtheorem{coro}[theo]{Corollary}

\newtheorem{cl}[theo]{Claim}
\theoremstyle{definition}
\newtheorem{defi}[theo]{Definition}
\newtheorem{rem}[theo]{Remark}
\newtheorem{obs}[theo]{Observation}

\newcommand{\phs}{$d$-{\sc Promise-Hitting-Set}}

\newcommand{\npphs}{$d$-{\sc Hitting-Set}}
\newcommand{\dphs}{$d$-{\sc Decision-Hitting-Set}}

\newcommand{\bis}{{\sc BIS}}

\newcommand{\gpis}{{\sc GPIS}}
\newcommand{\gpise}{{\sc GPISE}}

\newcommand{\defproblem}[3]{
  \vspace{1mm}
\noindent\fbox{
  \begin{minipage}{0.96\textwidth}
  \begin{tabular*}{\textwidth}{@{\extracolsep{\fill}}lr} #1 \\ \end{tabular*}
  {\bf{Input:}} #2  \\
  {\bf{Output:}} #3
  \end{minipage}
  }
  \vspace{1mm}
}

\renewcommand{\hat}{\widehat}
\renewcommand{\tilde}{\widetilde}

\title{
      Almost optimal query algorithm for hitting set\\
      using a subset query \footnote{A preliminary version of the work appeared in ISAAC 2018.}
      }
\author{
Arijit Bishnu
\footnote{
Indian Statistical Institute, Kolkata, India
}
\and
Arijit Ghosh
\footnotemark[1]
\and
Sudeshna Kolay
\footnote{Indian Institute of Technology Kharagpur, India}
\and
Gopinath Mishra
\footnote{University of Warwick, UK. Research supported in part by the Centre for Discrete Mathematics and its Applications (DIMAP), by EPSRC award EP/V01305X/1.}
\and 
Saket Saurabh
\footnote{The Institute of Mathematical Sciences, HBNI, Chennai, India}
}

\date{}

\begin{document}

\maketitle

\begin{abstract}

In this paper, we focus on {\sc Hitting-Set}, a fundamental problem in combinatorial optimization, through the lens of sublinear time algorithms. Given access to the hypergraph through a subset query oracle in the query model, we give sublinear time algorithms for {\sc Hitting-Set} with almost tight parameterized query complexity. In \emph{parameterized query complexity}, we estimate the number of queries to the oracle based on the parameter $k$, the size of the {\sc Hitting-Set}.
The subset query oracle we use in this paper is called Generalized $d$-partite Independent Set query oracle (GPIS) and it was introduced by Bishnu et al.~(ISAAC'18). GPIS is
a generalization to hypergraphs of the Bipartite Independent Set query oracle (BIS) introduced by Beame et al.~(ITCS'18 and TALG'20) for estimating the number of edges in graphs. Since its introduction GPIS query oracle has been used for estimating the number of hyperedges independently by Dell et al.~(SODA'20 and SICOMP'22) and Bhattacharya et al.~(STACS'22), and for estimating the number of triangles in a graph by Bhattacharya et al.~(ISAAC'19 and TOCS'21).
Formally, GPIS is defined as follows: 
\begin{center}
  {\em \gpis{} oracle for a $d$-uniform hypergraph $\mathcal{H}$ takes as input $d$ pairwise disjoint non-empty \\
  subsets $A_1, \ldots, A_d$ of vertices in $\cal H$ and answers whether there is a\\ 
  hyperedge in $\mathcal{H}$ that intersects each set $A_i$, where $i \in \{1, \, 2, \, \ldots, d\}$. }
\end{center}
For $d=2$, the \gpis{} oracle is nothing but \bis{} oracle.

We show that \npphs{}, the hitting set problem for $d$-uniform hypergraphs, can be solved using $\tOh_d(k^{d} \log n)$ \gpis{} queries. Additionally, we also showed that \dphs{}, the decision version of \npphs{} can be solved with $\tOh_d\left( \min \left\{ k^d\log n, k^{2d^2} \right\} \right)$ \gpis{} queries.  We complement these parameterized upper bounds with an almost matching parameterized lower bound that states that any algorithm that solves \dphs{} requires $\Omega \left( \binom{k+d}{d} \right)$ \gpis{} queries.

\end{abstract}

\section{Introduction}\label{sec:intro}

\noindent
In query complexity models for graph problems, the aim is to design algorithms that have access to the vertices $V(G)$ of a graph $G$, but not the edge set $E(G)$. Instead, these algorithms construct local copies of the graph by using oracles to probe $G$ and infer about a property of a part of the graph. Due to the lack of knowledge about the edges of the graph, often it is difficult to design algorithms even for problems that are classically known to have polynomial time algorithms. 

A natural optimization question in this model is to minimize the number of queries made on a worst case input to the oracle to solve the problem at hand. This has spawned the field of {\em query complexity}. The query complexity of an algorithm is the number of queries made to the oracle. Keeping this in mind, several query models have been designed through the years that strike a balance between not revealing too much information and revealing enough information per query to reduce the number of queries to solve a particular problem. 

There is a vast literature available on the query complexity of problems with classical polynomial time algorithms (refer to book~\cite{Oded17}). There have also been works that look at algorithmically hard problems through the lens of query complexity~\cite{IndykMRVY18, IwamaY18, OnakRRR12}. In this paper, we use ideas of parameterized complexity in order to study the query complexity of an NP-hard problem. The {\sc Hitting Set} (and {\sc Vertex Cover}) problem is a test problem for all new techniques of parameterized complexity and also in every subarea that parameterized complexity has explored. We continue 
this tradition and study the query complexity of {\sc Hitting Set}. Our query model is a generalization of a recently introduced query model by Beame et al.~\cite{talg/BeameHRRS20}.


\subsection{The query model}\label{subsec:model}

Given a graph $G$, $V(G)$ and $E(G)$ denote the vertex and edge set, respectively. For an edge $e\in E(G)$ with endpoints $u,v\in V(G)$, we denote $e=(u,v)$. Given a hypergraph $\mathcal{H}$, the vertex set and hyperedge sets are denoted by  $U(\mathcal{H})$ and $\cF(\cH)$, respectively. 
A $d$-uniform hypergraph has exactly $d$ vertices in all its hyperedges.
The set $\{1,2,\ldots,n\}$ is denoted by $[n]$. The vertex set, be it for graphs and hypergraphs, has a cardinality of $n$. Given a finite set $A$, $|A|$ will denote its cardinality. 
For a function $f(k)$, the set of functions of the form $\Oh(f(k)\cdot \log^{c} k)$ where $c$ is an absolute constant, is denoted by $\tOh(f(k))$. Similarly, the set of functions of the form $\tOh\left(c_{d} f(k)\right)$ where $c_{d}$ is a function of $d$, will be denoted by $\tOh_{d}(f(k))$.

Our goal in this paper is to look at the parameterized query complexity of the hitting set problem with oracle access. The specific oracle access we use is known as \gpis{} and was introduced by Bishnu et al.~\cite{DBLP:conf/isaac/BishnuGKM018}. Later, the same oracle was reintroduced by Dell et al.~\cite{DBLP:journals/siamcomp/DellLM22} with the name {\em Colorful Independent Oracle}. The precursor of these query oracle is the \bis{} oracle introduced by Beame et al.~\cite{talg/BeameHRRS20}. We start by defining the \bis{} oracle.  
\begin{description}
\item[\emph{Bipartite independent set oracle} (\bis)~\cite{talg/BeameHRRS20}.] For a graph $G$, given two disjoint non-empty subsets $A,B \subseteq V(G)$ as input, a \bis{} query oracle answers whether there exists an edge $(u,v) \in E(G)$ such that $u \in A$ and $v \in B$.
\end{description}
Given two vertices $u, \, v \in V(G)$, the often used \emph{edge existence query}~\cite{Oded17} asks for an yes/no answer to the question whether there exists an edge between $u$ and $v$. The \bis{} oracle, proposed by Beame et al.~\cite{talg/BeameHRRS20}, is a generalization over the \emph{edge existence query} in the sense that it asks for the existence of an edge between two disjoint sets of vertices. \bis{} was used to estimate the number of edges in a graph in~\cite{talg/BeameHRRS20}. The following oracle is a generalization of \bis{} to the hypergraph setting.
\begin{description}
\item[\emph{Generalized d-partite independent set oracle} (\gpis)~\cite{DBLP:conf/isaac/BishnuGKM018}.] For a $d$-uniform hypergraph $\mathcal{H}$, given $d$  pairwise disjoint non-empty subsets $A_1,A_2,\ldots,A_d \subseteq U(\mathcal{H})$ as input, a \gpis{} query oracle  answers whether there exists a hyperedge $(u_1,\ldots,u_d) \in \mathcal{F}(\mathcal{H})$ such that $u_i \in A_i$, for each $i \in [d]$.
\end{description}
We will get back \bis{} oracle when we substitute $d=2$ in the above definition.

Queries like \emph{degree query}, \emph{edge existence query}, \emph{neighbor query} (see~\cite{Oded17}), that obtain local 
information about the graph have its limitation in terms of not being able to achieve \emph{efficient} 
query costs. For example, the \emph{edge estimation problem}, where the objective is to estimate the number of edges in the graph, has \emph{linear}\footnote{linear in the number of vertices.} query complexity in the worst case~\cite{Feige06,GoldreichR08}. It is a natural question that whether we can solve the problem at hand efficiently with a query access that have more power in the 
sense that it goes beyond obtaining local information and generalizes earlier queries. Beame et al.~\cite{talg/BeameHRRS20} introduced \bis{} query model and gave an algorithm for the edge estimation problem using polylogarithmic \bis{} queries. To get a  better motivation behind \bis{} query, please refer to~\cite{talg/BeameHRRS20}. The independent set based oracles like \bis{} and \gpis{} fall under the category of \emph{group/subset queries}, which was introduced by Stockmeyer~\cite{Stockmeyer85,Stockmeyer83} and formalized by Ron and Tsur~\cite{RonT16}. 

\paragraph*{Use of independent set oracles.} Independent set based oracles mostly report on the intersection of the edge set with set(s) of vertices -- the oracles give a YES/NO answer to the existence of an intersection, in a few cases they even count the number of such intersections. Lately there has been a wide range of interest in them. By now, they have been used for solving a lot of problems -- edge and hyperedge estimation in graphs and hypergraphs~\cite{talg/BeameHRRS20,soda/0001LW20,DBLP:journals/mst/BhattacharyaBGM21,DBLP:journals/siamcomp/DellLM22,abs-1908-04196}, sampling edges and hyperedges~\cite{soda/0001LW20,DBLP:journals/siamcomp/DellLM22}, fine-grained complexity of approximate counting problems~\cite{DellL21}, computing minimum cut~\cite{RubinsteinSW18} and submodular function minimization~\cite{itcs/GraurPRW20} using {\sc Cut} queries; in the {\sc Cut} query for a graph $G$, the oracle takes as input a vertex subset $S \subseteq V(G)$ and outputs the number of cut edges between $S$ and $V(G) \setminus S$.

In order to understand the limitations of independent set based oracles in terms of NP-Hard problems, it is reasonable to study query complexity of their parameterized versions. As {\sc Hitting Set} has been a kind of a test problem for any new area/technique that parameterized complexity has explored, we will focus on the parameterized decision (optimization) version of {\sc Hitting Set} using \gpis{} oracle. Note that Iwama and Yoshida~\cite{IwamaY18} initiated the study of parameterized version of some NP-Hard problems in the graph property testing framework with access to standard oracles, like degree query and neighbor query along with some \emph{added power} to the oracle. We will give the details of their work in Section~\ref{sec:rel}. We believe that apart from the oracles used in~\cite{IwamaY18}, these independent set based query models will be useful to study the (parameterized)  query complexity of other NP-Hard problems.

\paragraph{Efficient implementation of the query oracle.} {\sc GPIS} queries for graphs, that is {\sc BIS} queries, is a special case of the {\em vector-matrix-vector} query supported on the adjacency matrix of the graph. Given a {\em vector-matrix-vector} query access to an unknown matrix $A \in \R^{n\times n}$, {\em vector-matrix-vector} query will output the value of $x^{T} A y$ for any two specified vectors $x$ and $y$ in $\R^{n}$. Recently, {\em vector-matrix-vector} query and similar linear algebraic queries have been used to study properties of unknown matrices~\cite{DBLP:conf/approx/RashtchianWZ20, SunWYZ19,DBLP:conf/approx/BishnuGM21,BakshiCW22,NeedellSW22,BishnuGMP22}. Observe that {\sc BIS} queries, and more generally {\sc GPIS} queries, can be implemented using one {\em vector-matrix-vector} query on the adjacency matrix of the graph. On the implementation side observe that these new linear algebraic queries can be efficiently implemented using specialized hardware or in distributed environments where inner product between two vectors can be implemented efficiently. 
\color{black}

\subsection{Problem definition and our results}
\label{sec:probdefi}

In our framework, the vertices of the hypergraph are known while the hyperedges are unknown, and we have {\sc GPIS} oracle access to the hypergraph.  Broadly the idea is to make specific queries to {\sc GPIS} oracle and use the outcomes of these queries to build a {\em reduced hypergraph instance}. We will then show that for solving the original problem we only need to run traditional (FPT) algorithms on this reduced instance. While stating our results, we will only care about the number of queries required to solve the problem. Our main goal is to understand the query complexity in terms of the input parameters of the problem. Observe that our bounds on the query complexity are not directly comparable with the time complexities of FPT algorithms.

Our algorithms will use the technique of {\em color coding}~\cite{DBLP:conf/stoc/AlonYZ94,DBLP:journals/jacm/AlonYZ95,saketbook15}, and {\em stability} of sunflowers under random sampling. For our lower bounds, we will use the communication complexity framework developed by Eden and Rosenbaum~\cite{DBLP:conf/approx/EdenR18}. 

\paragraph{\npphs{} problem.}
The \npphs{} problem is defined as follows.

\defproblem{\npphs}{The set of vertices $U(\mathcal{H})$ of a $d$-uniform
hypergraph $\mathcal{H}$, access to a \gpis{} oracle, and a positive integer $k$.}{
Among all the subsets of $U(\cH)$ of size at most $k$, output a set $S$ of smallest size such that any hyperedge of $\cH$ intersects with $S$. Otherwise, we report no
such set exist among all the subsets of $U(\cH)$ of size at most $k$.}

Note that, in this paper, we consider $d$ as a constant independent of $k$. \dphs{} is the usual decision version of \npphs{}. The main results of our work are as follows; they include both upper and matching lower bounds for the {\sc Hitting-Set} problem.
\begin{theo}[Upper bounds]
\label{theo:hs_np_bise}
\begin{itemize}

\item[(i)]
    \npphs{} can be solved using $\tOh_d(k^{d} \log n)$  \gpis{} queries with high probability. 

\item[(ii)] 
    \dphs{} can be solved using $\tOh_d\left( \min \left\{ k^d\log n, k^{2d^2} \right\} \right)$ \gpis{} queries with high probability.
\end{itemize}
\end{theo}

\begin{theo}[Lower bound]
\label{theo:lowerbound}
 Any algorithm that solves \dphs{}, with probability at least $2/3$, requires $\Omega \left( \binom{k+d}{d} \right)$ \gpis{} queries.
\end{theo}

\subsection{Related Works}
\label{sec:rel}

\noindent
To the best of our knowledge, the only work prior to ours related to parameterization in the query complexity model was by Iwama and Yoshida~\cite{IwamaY18}. They studied property testing for several parameterized NP optimization problems in the query complexity model. For the query, they could ask for the degree of a vertex, neighbors of a vertex -- both \emph{local queries} and had an added power of sampling an edge uniformly at random. As the probability space is over the entire edge set, asking for a random edge does not qualify to be a local query.
To justify the added power of the oracle to sample edges 
uniformly at random, they have shown that $\Omega(\sqrt{n})$ degree and neighbor queries are required to solve {\sc 
Vertex-Cover}. Apart from that, an important assumption in their work is that the algorithms knew the number of edges, 
which is not what is usually done in query complexity models. Under these assumptions, they study the parameterized query complexity of vertex cover, feedback vertex set, multicut, dominating set and non-existence of paths of specific length and give constant query testable algorithms if the parameter $k$ is treated to be a constant.

Note that our query oracles can access some global information. However, our oracles  do not use  any randomness, does not know the number of edges,  and have a simple unifying structure in terms of asking for the existence of an edge between disjoint sets of vertices. We feel that our work marked by its use of independent set based oracle queries is not comparable to the work by Iwama and Yoshida~\cite{IwamaY18}. We mention in passing that their vertex cover algorithm admits a query complexity of $\tOh(\frac{k 2^k}{\eps^2})$ and either finds a vertex cover of size at most $k$ or decides that there is no vertex cover of size bounded by $k$ even if we delete $\eps m$ 
edges, where the number of edges $m$ is known in advance. In contrast to the work of Iwama and Yoshida~\cite{IwamaY18}, our algorithm uses \bis{} query for the vertex cover problem, consider all instances, and it neither knows the number of edges in the graph in advance nor does it estimate the number of edges in the graph. The query complexity of our algorithm is $\tOh\left(k^{2}\log n\right)$ and we either find a vertex cover of size at most $k$ if it exists or decide that there is no vertex cover of size bounded by $k$. We also prove an almost matching lower bound for the problem.

\subsection*{Organization of the paper}

We discuss preliminaries in Section~\ref{sec:prelim}. In Section~\ref{sec:hs}, we deal with the promise version of the problem leading to the hitting set problem. The decision version of the hitting set problem is discussed in Section~\ref{sec:hs_nonprom}. Section~\ref{sec:lowerbound-restate} has the details of the lower bound and Section~\ref{sec:cut} concludes the paper.

\section{Preliminaries}\label{sec:prelim}
\subsection{Notations and definitions}

\noindent
A \emph{hypergraph} is a set system $(U(\cH),\cF(\cH))$, where $U(\cH)$ is the set of vertices and $\cF(\cH)$ is the set of hyperedges. A hypergraph $\mathcal{H}'$ is a \emph{sub-hypergraph} of $\mathcal{H}$ if $U(\mathcal{H}')\subseteq U(\mathcal{H})$ and $\mathcal{F}(\mathcal{H}')\subseteq \mathcal{F}(\mathcal{H})$. For a hyperedge $F \in \mathcal{F}(\cH)$, $U(F)$ or simply $F$ denotes the subset of vertices that form the hyperedge. Given hypergraphs $\cH_1,\, \cH_2$ defined on the same set of vertices, the hypergraph $\cH_1 \cup \cH_2$ is such that $U(\cH_1\cup \cH_2) = U(\cH_1) = U(\cH_2)$ and $\cF(\cH_1\cup \cH_2) = \cF (\cH_1) \cup \cF(\cH_2)$.
All hyperedges of a \emph{$d$-uniform hypergraph} have exactly $d$ vertices. $\hs(\cH)$ denotes a minimum $d$-{\sc Hitting Set} of the $d$-uniform hypergraph $\cH$. A \emph{packing} in a hypergraph $\cH$ is a family $\mathcal{F}'$ of hyperedges such that for any two hyperedges $F_1,F_2 \in \mathcal{F}'$, $U(F_1)\cap U(F_2) = \emptyset$.

For us ``choose a random hash function $h:V \rightarrow [N]$'', means that each vertex  in $V$ is colored with one of the $N$ colors uniformly and independently at random.

In this paper, for a problem instance $(I,k)$ of a parameterized problem $\Pi$, a \emph{high probability} event means that it occurs with probability at least $1- \frac{1}{k^c}$, where $k$ is the given parameter and $c$ is a constant. The following observation is important for the analysis of algorithms described in this paper.

\begin{obs}\label{obs:const-prob-boost}
\begin{itemize}
 \item[(i)]  Let $\Pi$ be a parameterized maximization~(minimization) problem and let $(I,k)$ be an instance of $\Pi$. Let $\mathcal{A}$ be a randomized algorithm for $\Pi$, with success probability at least $p$, where $ 0 < p < 1$ is a constant. Then, if we repeat 
 $\mathcal{A}$ for $C\log{k}$ times for a suitably large constant $C$ and report the maximum~(minimum) sized output over  $C \log{k}$ outcomes, 
 then the event that $\mathcal{A}$ succeeds occurs with high probability. If the query complexity of algorithm $\mathcal{A}$ is $q$, then the query complexity of the $C\log{k}$ repetitions of $\mathcal{A}$ is $\tOh(q)$.
 \item[(ii)] Let $\Pi$ be a parameterized decision problem and let $(I,k)$ be an instance of $\Pi$. Let $\mathcal{A}$ be a randomized algorithm for $\Pi$, with success probability at least $p$, where $\frac{1}{2} < p < 1$ is a constant. Then, if we repeat 
 $\mathcal{A}$ for $C\log{k}$ times for a suitably large constant $C$ and report the \emph{majority} of the $C \log{k}$ outcomes, 
 then the event that $\mathcal{A}$ succeeds occurs with high probability. If the query complexity of algorithm $\mathcal{A}$ is $q$, then the query complexity of the $C\log{k}$ repetitions of $\mathcal{A}$ is $\tOh(q)$.
\end{itemize} 
\end{obs}

\paragraph*{Representative set:}   

Let $\cH$ be a hypergraph. $\cF' \subseteq \cF(\cH)$ is said to be a \emph{$k$-representative set} corresponding to $\cH$ if the following is satisfied for any $X \subset U(\cH)$ of size $k$. If there is an $F \in \cF(\cH)$ satisfying $X \cap F = \emptyset$, then there exists $F' \in \cF'$ such that $X \cap F' = \emptyset$.

The following proposition gives a bound on the size of a $k$-representative set corresponding to a $d$-uniform hypergraph.
\begin{pre}[\cite{BOLLOBAS65}]\label{pre:rep-set-system}
If $\cH$ is $d$-uniform hypergraph, then there exists a ${k+d} \choose {d}$ size $k$-representative set corresponding to $\cH$.
\end{pre}

\begin{coro}[\cite{saketbook15}]\label{coro:rep-set-system}
For a set system $\mathcal{H}$ as above, consider the family $\mathcal{Z} = \{U(F)~\vert~F \in \mathcal{F}(\mathcal{H})\}$ and let $\hat{\mathcal{Z}}$ be a $k$-representative set of $\mathcal{Z}$ as obtained in Proposition~\ref{pre:rep-set-system}. Let ${\mathcal{H}}'$ be the set system where $U({\mathcal{H}}')=\bigcup_{Z \in \hat{\mathcal{Z}}} Z$ and $\mathcal{F}({\mathcal{H}}') = \{F\in\mathcal{F}~\vert~ U(F) \in \hat{\mathcal{Z}}\}$. $(\mathcal{H},k)$ is a YES instance of \dphs{} if and only if $({\mathcal{H}}',k)$ is a YES instance of \dphs.
\end{coro}

\subsection{Technical preliminary}\label{sec:gpis-gpise}

For ease of exposition, we now define a related query oracle \gpise{} that returns a witness hyperedge for a YES answer of \gpis{} and returns NULL, otherwise.  The formal definition of \gpise{} is as follows.
\begin{description}
\item[\emph{Generalized d-partite independent set edge oracle} (\gpise):] For a $d$-uniform hypergraph $\mathcal{H}$, given $d$  pairwise non-empty disjoint subsets $A_1,A_2,\ldots,A_d \subseteq U(\mathcal{H})$ as input, a \gpise{} query oracle outputs a hyperedge $(u_1,\ldots,u_d) \in \mathcal{F}(\mathcal{H})$ such that $u_i \in A_i$, for each $i \in [d]$; otherwise, the \gpise{} oracle reports NULL.
\end{description}

The following observation says that \gpise{} is equivalent to \gpis{} upto $\Oh_d(\log n)$ factor, that is, a \gpise{} query can be simulated by using $\Oh_d(\log n)$ \gpis{} queries deterministically. The intuition is as follows: Let $A_1,\ldots, A_d$ be the input to  \gpise{} query. The idea is to make a \gpis{} query with input  $A_1,\ldots, A_d$. If the answer is {\sc No}, then the answer is same to the \gpise{} query. If the answer is {\sc Yes}, then we roughly halve each $A_i$ say $A_{i1}$ and $A_{i2}$, make $2^d$ \gpis{} query, and recurse suitably.
\begin{obs}\label{obs:gpis-gpise}
Let $A_1, \ldots, A_d$ be $d$ pairwise disjoint subsets of $U(\cH)$.
A \gpise{} query with input $A_1,\ldots,A_d$ can be simulated by using $\Oh_d(\log n)$ \gpis{} queries.
\end{obs}
\begin{proof}
We describe the simulation process in a recursive fashion. 
We first make a \gpis{} query with input $A_1,\ldots,A_d$. If \gpis{} reports there is no hyperedge spanning the sets $A_1,\ldots,A_d$, then we report NULL as the answer to the \gpise{} query. Otherwise, for each $i\in [d]$, we partition each $A_i$ into two parts, that is, $A_{i1}$ and $A_{i2}$ such that $\size{A_{i1}}=\lceil \frac{\size{A_{i}}}{2} \rceil$ and  $\size{A_{i2}}=\lfloor \frac{\size{A_{i}}}{2}\rfloor$. For each $A_{1j_{1}}, \ldots,A_{dj_{d}}$ with $j_{i} \in \{1,2\}$ and $i \in [d]$, we make a \gpis{} query with input $A_{1j_{1}}, \ldots,A_{dj_{d}}$. Note that we make $2^d$ \gpis{} queries. Observe that there exists at least one combination of $A_{1j_{1}}, \ldots,A_{dj_{d}}$ such that \gpis{} report that $m(A_{1j_{1}}, \ldots,A_{dj_{d}})\neq 0.$~\footnote{$m(A_{1j_{1}}, \ldots,A_{dj_{d}})$ is the number of hyperedges having a vertex in $A_{ij_{i}}$'s for each $i\in [d]$} Now we call for \gpise{} query with one such $A_{1j_{1}}, \ldots,A_{dj}$ (such that $m(A_{1j_{1}}, \ldots,A_{dj_{d}})\neq 0$) as input, and reports the answer of the \gpise{} query with input $A_{1j_{1}}, \ldots,A_{dj_{d}}$ as the answer to the \gpise{} query with input $A_1,\ldots,A_d$. The correctness of the answer to the \gpise{} query follows 
from the description of the simulation process. Let $Q_E(A_1,\ldots,A_d)$ denotes the number of \gpis{} query, that our simulation process makes, to answer \gpise{} query with input $A_1,\ldots,A_d$. Hence,

\begin{eqnarray*}
Q_E(A_1,\ldots,A_d) &\leq& 1+2^d+\max\limits_{A_{1j_{1}}, \ldots,A_{dj_{d}}}Q_E\left(A_{1j_{1}}, \ldots,A_{dj_{d}}\right).
\end{eqnarray*}
Observe that $Q_E(A_1,\ldots,A_d)=\Oh_d(\log n)$.
\end{proof}

\begin{obs}\label{obs:find-an-edge}
Let $\cG$ be subgraph, of a $d$ uniform hypergraph $\cH$, induced by $V \subseteq U(\cH)$. There exists an algorithm $\cA$ that makes $\Oh_d \left(\log \frac{1}{\delta}\right)$ \gpise{} queries and outputs either a hyperedge or {\sc Null} with the following guarantee: if there exists at least one hyperedge in $\cG$, then $\cA$ returns a hyperedge in $\cG$ with probability $1-\delta$; otherwise, $\cA$ reports {\sc Null}.
\end{obs}
\begin{rem}
By Observation~\ref{obs:gpis-gpise}, the above algorithm $\cA$ implies an algorithm that uses $\Oh_d\left(\log n \log \frac{1}{\delta}\right)$ \gpis{} queries and gives an output that is same as that of $\cA$.
\end{rem}
\begin{proof}[Proof of Observation~\ref{obs:find-an-edge}]
We use color coding technique here. Let us consider partitioning the vertex set $V$ into $d$ parts $B_1,\ldots,B_d$ such that each vertex in $V$ is present in one of the $B_i$s uniformly at random, and make a \gpise{} query with input $B_1,\ldots,B_d$. The algorithm $\cA$ repeats the above procedure $R=\Oh\left({\frac{d!}{d^d}}\log \frac{1}{\delta}\right)$ times, and  reports a hyperedge if at least one of the $R$ \gpise{} queries reports a hyperedge. Otherwise, $\cA$ reports that there is no hyperedge in $\cG$.
 
 The query complexity of $\cA$ follows from the description. Let us prove the correctness now. If there is no hyperedge in $\cG$, then all of the $R$ \gpise{} queries reports NULL. Now consider the case when there exists at least one hyperedge in $\cG$ and a particular \gpise{} query made by the algorithm $\cA$ with input $B_1,\ldots,B_d$.  The probability that all of the $d$ vertices of the particular hyperedge are in different $B_i$s is $\frac{d!}{d^{d}}$ and the \gpise{} query reports such an edge with probability at least $\frac{d!}{d^{d}}$. As we are making $R$ \gpise{} queries, we report a hyperedge with probability at least $1-\left(1-\frac{d!}{d^{d}}\right)^R\geq 1-\delta$.

\color{black}
\end{proof}

\color{black}

\section{Algorithm for \npphs}
\label{sec:hs}

\noindent
We will now prove the following result. 
\begin{theo}[Restatement of Theorem~\ref{theo:hs_np_bise}~(i) in terms of \gpise{} queries]
\label{theo:hs_np_bise-restate}
\npphs{} can be solved with $\tOh_d \left( k^{d} \right)$ \gpise{} queries.
\end{theo}
\noindent
Observe that the above theorem assumes access to \gpise{} (instead of \gpis{}) query, and Theorem~\ref{theo:hs_np_bise}~(i) now directly follows from Observation~\ref{obs:gpis-gpise}.

The algorithm for \npphs{} will use an algorithm  admitting a query complexity of $\tOh(k^d)$ for a promise version of this problem where the input instance is promise to have a hitting set of size at most $dk$. 

The main idea to solve the promise version is to sample a \emph{suitable} sub-hypergraph having $\widetilde{O}_d\left(k^d\right)$ hyperedges, using \gpise{} queries on the input hypergraph, such that the hitting set of the sampled hypergraph is a hitting set of the original hypergraph and vice versa. Two main ingredients in the proof of Theorem~\ref{theo:hs_np_bise-restate} are the following:

\begin{itemize}
    \item[1.]
        Structure of a \emph{sunflower} in a hypergraph~\cite{sunflower}: We use the following structural result of a  a $d$-uniform hypergraph $\cH$. If the number of hyperedges in $\cH$ is at least $d!k^d$, then there exists a $(k+1)$-sunflower in $\mathcal{H}$ (see Definition~\ref{defi:sunflower} and Proposition~\ref{pre:sunflower}). This helps us to design an algorithm for the promise version.
    
    \item[2.]
        An algorithm for {\sc Gap-$d$-Hitting-Set} problem using
        $\tOh \left( k\right)$ GPISE queries: here the algorithm distinguishes between the case when the hitting set is at most $k$ and at least $dk$.  
\end{itemize}
\color{black}
The $d$-{\sc Hitting Set} problem can be solved by using the algorithm for the promise version of the $d$-{\sc Hitting Set} problem along with the algorithm for {\sc Gap-$d$-Hitting-Set} problem.

\subsection{{\sc Gap-$d$-Hitting-Set} problem}\label{sec:gap}
\noindent
  In {\sc Gap-$d$-Hitting-Set} on a $d$-uniform hypergraph $\cH$, the objective is to report {\sc Accept} if $\cH$ has a hitting set of a size at most $k$, to report {\sc Reject} if the size of any minimum hitting set of $\cH$ is more than $dk$, and to report {\sc Accept} or {\sc Reject} arbitrarily if the hitting set lies between $k$ and $dk$. We will show (in Observation~\ref{obs:gappacking}) that {\sc Gap-$d$-Hitting-Set} can be solved by using $\tOh(k)$ \gpise{} queries. For the $d$-{\sc Hitting Set} problem, we first solve {\sc Gap-$d$-Hitting-Set}. If the algorithm for {\sc Gap-$d$-Hitting-Set} reports {\sc Reject}, then we conclude that the size of the minimum hitting set of $\cH$ is at least $k$. If algorithm for {\sc Gap-$d$-Hitting-Set} reports {\sc Accept}, then $\cH$ has a hitting set of size at most $dk$. Now we can use our algorithm for the promise version of $d$-{\sc Hitting Set} to give the final answer to the non-promise $d$-{\sc Hitting Set}.
  
  \begin{obs}\label{obs:gappacking}
  {\sc Gap-$d$-Hitting-Set} can be solved by using $\tOh(k)$ \gpise{} queries.
  \end{obs}
  \begin{proof}
  We find a \emph{packing}  of size at most $k+1$ in a greedy fashion, by using $\tOh(k)$ \gpise{} queries as follows. 
  \begin{itemize}
  \item[(i)] Set $V=U(\cH)$, $\cG=\cH$.
  \item[(ii)] Run algorithm $\cA$ (the algorithm corresponding to Observation~\ref{obs:find-an-edge}) on $\cG$ with parameter $\delta=\frac{1}{k^c}$, where $c$ is a suitably large constant larger  than $1$.
  \item[(iii)] If $\cA$ reports that there is no edge in $\cG$, then report {\sc Accept} and {\sc Quit}. 
  \item[(iv)] Let $F$ be the hyperedge in $\cG$ that is returned by $\cA$. If we have seen $k+1$ hyperedges (including $F$), then we report {\sc Reject} and {\sc Quit}. 
  \item[(v)] Otherwise, we delete all the vertices in $F$ from $\cG$, that is, we set $\cG=\cG \setminus F$. Go to Step (ii). 
  \end{itemize}
  The above algorithm calls algorithm $\cA$ with parameter $\delta=\frac{1}{k^c}$ at most $k+1$ times. From Observation~\ref{obs:find-an-edge}, each call of algorithm $\cA$ requires $\Oh_d(\log k)$ \gpise{} queries and succeeds with probability at least $1-\delta=1-\frac{1}{k^c}$. So, the above algorithm for {\sc Gap-$d$-Hitting-Set} makes $\tOh(k)$ \gpise{} queries and succeeds with probability at least $1-\frac{1}{k^{c-1}}$. 
  Now we discuss the correctness proof of our algorithm for  {\sc Gap-$d$-Hitting-Set} assuming all calls to algorithm $\cA$ suceed. Observe that our algorithm (for {\sc Gap-$d$-Hitting-Set}) finds a packing  of size at most $k+1$. Moreover, if the algorithm stops after finding a packing of size at most $k$, then those set of at most $k$ hyperedges correspond to a maximal packing. If hypergraph $\cH$ has a hitting set of size at most $k$, then the size of any (maximal) packing is at most $k$. In this case, our algorithm quits after finding at most $k$ hyperedges that correspond to a maximal packing, and we report {\sc Accept}. Now, if the size of the minimum hitting set  of $\cH$ is more than $dk$, then the size of any maximal packing is at least $k+1$. In this case, our algorithm will be able to find a packing of size at least $k+1$, and we report {\sc Reject}. 
  \end{proof}


\subsection{Algorithm for {\npphs} via \phs}\label{sec:hsprom}
\noindent
In this Section, we begin by studying the following promise problem.

\defproblem{\phs}{Parameter $k \in \N$, the set of vertices $U(\mathcal{H})$ of a $d$-uniform hypergraph $\mathcal{H}$ such that $ \size{\hs(\cH)} \leq k$, and the access to a \gpise{} oracle.}{
Among all the subsets of $U(\cH)$ of size at most $k$, output a set $S$ of smallest size such that any hyperedge of $\cH$ intersects with $S$.
}

We will show at the end of this section that the algorithm for \npphs{} follows from the algorithms for {\sc Gap-$d$-Hitting-Set} and {\phs} problems. Note that we have discussed about {\sc Gap-$d$-Hitting-Set} in Section~\ref{sec:gap}. The following theorem formally states the result on \phs{}.


\begin{algorithm}[!hbt]
\caption{Algorithm for \phs{}}
\label{algo:phs}
\KwIn{Parameter $k \in \N$, the set of vertices $U(\mathcal{H})$ of a $d$-uniform hypergraph $\mathcal{H}$ such that $ \size{\hs(\cH)} \leq k$, and the access to a \gpise{} oracle.}
\KwOut{A minimum hitting set of $\mathcal{H}$ that is of size at most $k$.}
\Begin
	{
Take $\alpha \log k$ random hash functions of the form $h:U(\cH)\rightarrow [\beta k]$, where $\alpha =100d^2$ and $\beta=100d^32^{d+5}$.\\
\For{ (each hash function $h$)}
{
Find $U_i=\{u \in U(\mathcal{H}) : h(u) =i\}$, where $i \in [\beta k]$.\\
 Make a \gpise{} query with input $(U_{i_1},\ldots,U_{i_d})$ for each $1 \leq i_1< \ldots < i_d \leq \beta k$ such that $U_{i_j} \neq \emptyset$ $\forall j \in [d]$.\\
Let $\mathcal{F}'$ be the set of hyperedges that are output by the $\Oh(k^d)$ \gpise{} queries. \\
 Generate a subhypergraph 
${\mathcal{H}}^h$ of $\mathcal{H}$ such that 
$U(\mathcal{H}^h)=U(\mathcal{H})$ and $\mathcal{F}(\mathcal{H}^h)=\cF'$.
}
Let $\mathcal{H}_1,\ldots,\mathcal{H}_{\alpha  \log k}$ be the subhypergraphs generated by $\alpha \log k$ hash functions.\\
 Find $\hat{\mathcal{H}}=\mathcal{H}_1 \cup \cdots \cup \mathcal{H}_{\alpha \log k}$.\\
Report $\hs(\hat{\cH})$ as the output.

}
\end{algorithm}
\begin{theo}
\label{theo:hit_set}
\phs{} can be solved with $\tOh(k^{d})$ \gpise{} queries with high probability.
\end{theo}

Here, we give an outline of the algorithm for \phs. The algorithm is inspired by the streaming algorithm of ~\cite{ChitnisCEHMMV16}. The algorithm chooses $\alpha \log k$ hash functions of the form $h:U(\cH)\rightarrow [\beta k]$, where $\alpha$ and $\beta$ are suitable constants depending on $d$. Note that each hash function partitions the vertex set into $\beta k$ parts. We make \gpise{} queries between each $d$-combination of partition, and generate the subhypergraph with the hyperedges output by the \gpise{} queries. Let $\widehat{\cH}$ be the union of subhypergraphs generated due to the hash functions. The algorithm finally finds a minimum hitting set of $\widehat{\cH}$. The formal description of the algorithm is presented in Algorithm~\ref{algo:phs}. The following lemma establishes that any minimum hitting set of $\widehat{\cH}$ is also a minimum hitting set of $\cH$.

\begin{lem}[Proof in Section~\ref{ssec-proof-lem:hs_main}]
\label{lem:hs_main}
If $\size{\hs (\mathcal{H})} \leq k$, then $\hs (\mathcal{H})=\hs (\hat{\mathcal{H}})$ with high probability.
\end{lem}
\begin{rem}
{The statement of our Lemma~\ref{lem:hs_main} is same as that of~\cite[Theorem~3.2]{ChitnisCEHMMV16}, but the proof is not. We feel the proof of~\cite[Theorem~3.2]{ChitnisCEHMMV16} is incomplete. The authors argue that $\hs(G)=\hs(U \cup F)$ where $G$ denotes the hypergraph, $U$ is the set of large cores and $F$ is the set of hyperedges that do not include any significant core. Next, the authors argue that $\hs(U \cup F)=\hs(U' \cup F)$ where $U'$ is the set of large cores that do not contain  significant cores. We feel that their statement is correct but the part of the proof meant for this, is sketchy. This is mainly because $\hs(U' \cup F)$ may not hit some hyperedges in $G$ that contain a large core $C$ such that $C$ contains a significant core but not large core $C'$. We prove Lemma~\ref{lem:hs_main} in  Section~\ref{ssec-proof-lem:hs_main}.}
\end{rem}

Observe that we are done with the proof of Theorem~\ref{theo:hit_set}  from Lemma~\ref{lem:hs_main}, except the query complexity of \phs{}. The query complexity of $\tOh_{d}(k^d)$ follows from the description of Algorithm~\ref{algo:phs}.

We finally come to the proof of Theorem~\ref{theo:hs_np_bise} (restated as Theorem~\ref{theo:hs_np_bise-restate}).

\begin{proof}[Proof of Theorem~\ref{theo:hs_np_bise-restate}]
We first run the algorithm of {\sc Gap-$d$-Hitting-Set} that succeeds with high probability (see Observation~\ref{obs:gappacking}). Under the assumption that the algorithm of {\sc Gap-$d$-Hitting-Set} succeeds, it reports {\sc Accept} if $\cH$ has a hitting set of size at most $k$, reports {\sc Reject} if the size of any minimum hitting set of $\cH$ is more than $dk$, and it reports {\sc Accept} or {\sc Reject} arbitrarily if the hitting set is more than $k$ and at most $dk$.

If the algorithm of {\sc Gap-$d$-Hitting-Set} reports {\sc Reject}, we conclude that $\size{\hs(\cH)}\geq k+1$.  So, in this case we report that there does not exist any hitting set of size at most $k$.
 Otherwise, if the algorithm of {\sc Gap-$d$-Hitting-Set} reports {\sc Accept}, then $\size{\hs(\cH)} \leq dk$. As $\size{\hs(\cH)} \leq dk$, $\hs(\cH)$ can be found using our algorithm for \phs{} by making $\tOh((dk)^d)$ \gpise{} queries. If $\size{\hs(\cH)} \leq k$, we output $\hs(\cH)$ and if $\size{\hs(\cH)} > k$, we report that there does not exist a hitting set of size at most $k$. The total number of \gpise{} queries made by our algorithm for \npphs{} is $\tOh((dk)^{d})$.
\end{proof}

Only thing that is left to show is the proof of Lemma~\ref{lem:hs_main}.

\subsection{Proof of Lemma~\ref{lem:hs_main}}
\label{ssec-proof-lem:hs_main}
\noindent
To prove Lemma~\ref{lem:hs_main}, we need some intermediate definitions and results. As mentioned earlier, we use the structure of the sunflower in a hypergraph~\cite{sunflower}. The core of a sunflower is the pairwise intersection of the hyperedges present in the sunflower, which is formally defined as follows. 
\begin{defi}
\label{defi:sunflower}
Let $\mathcal{H}$ be a $d$-uniform hypergraph; $\cS=\{F_1,\ldots,F_r \}\subseteq \cF(\mathcal{H})$ is an \emph{r-sunflower} in $\mathcal{H}$ if there exists $C \subseteq U(\mathcal{H})$ such that $F_i \cap F_j = C$ for all $1 \leq i < j \leq r$. $C$ is defined to be the \emph{core} of the sunflower $\mathcal{S}$ and 
$\cP=\{F_i \setminus C : i \in [r]\}$ is defined as the set of \emph{petals} of the sunflower $\cS$ in $\mathcal{H}$.
\end{defi}

Based on the number of hyperedges forming the sunflower, the core of a sunflower can be \emph{large}, \emph{significant}, or \emph{small}. We will now fix the definitions in such a way that each large core is significant and each significant core (and hence, large core also) must intersect with any hitting set.

\begin{defi}[Sunflowers large and significant]
Let $S_{\mathcal{H}}(C)$ denote the maximum integer $r$ such that $C$ is the core of an $r$-sunflower in $\mathcal{H}$. If $S_{\mathcal{H}}(C) > 10dk$, $C$ is \emph{large}. If $S_{\mathcal{H}}(C)>k$, $C$ is 
\emph{significant}. $C$ is \emph{small} if it is not significant.
\end{defi}

The promise that the hitting set is bounded by $k$, will help us 
\begin{itemize}
\item[(i)] to bound the number of hyperedges that do not contain any large core as a subset, 

\item[(ii)] to guarantee that all the large cores in the original hypergraph, that do not contain any significant cores as a subset, are significant in the sampled hypergraph with high probability. This will ensure that the large cores in the original hypergraph will intersect any hitting set of the sampled hypergraph, and

\item[(iii)] to guarantee that all the hyperedges that do not contain any large core as a subset, are present in the sampled hypergraph with high probability. 
\end{itemize}
Using the above observations, we can prove that the hitting set of the sampled hypergraph is the hitting set of the original graph with high probability. To formalize the above discussion, we state the following proposition and then define some sets, which will be needed for our analysis.

\begin{pre}[\cite{sunflower}]
\label{pre:sunflower}
Let $\mathcal{H}$ be a $d$-uniform hypergraph. If $\size{\cF(\mathcal{H})} > d!k^d$, then there exists a $(k+1)$-sunflower in $\mathcal{H}$.
\end{pre}

\begin{defi}[Additional notations]
\label{defi:sets}
In the hypergraph $\mathcal{H}$,  
\begin{itemize}
    \item 
        $\cC$ is the set of \emph{large} cores;
    
    \item
        $\cF_s$ is the family of edges that do not contain any \emph{large} core;
    
    \item
        $\cC'$ is the family of \emph{large} cores none of which contain a \emph{significant} core as a subset.
\end{itemize}
\end{defi}

The following two results (Lemmas~\ref{lem:hs_bound} and~\ref{lem:hs_bound1}) give useful bounds with respect to the input instances of \phs.
\begin{lem}
\label{lem:hs_bound}
If $\size{\hs (\mathcal{H})} \leq k$, then $ \size{\cF_s} \leq d!   (10dk)^d$.  That is, if a hitting set of the hypergraph $\cH$ is bounded by $k$, then the number of hyperedges that do not contain any large core is  at most  $d!   (10dk)^d$.
\end{lem}
\begin{proof}
If $\size{\cF_{s}} > d!(10dk)^d$, then there exists a $(10dk +1)$-sunflower $\cS$ in $\cH$ by Proposition~\ref{pre:sunflower} such that each edge in $\cS$ belongs to $\cF_s$. First, since $\size{\hs(\cH)} \leq k$, the core $C_{\cH}(\cS)$ of $\cS$ must be non-empty. 
Note that $C_{\cH}(\cS)$ is a large core and $C_{\cH}({\cS})$ is contained in every edge in $\cS$. Observe that we arrived at a contradiction, because any edge in 
$\cS$ is also an edge in $\cF_{s}$, and by definition any edge in $\cF_{s}$ does not contain a large core by definition. Hence, $\size{\cF_{s}} \leq d!(10dk)^d$.
\end{proof}

\begin{lem}
\label{lem:hs_bound1}
If $\size{\hs (\mathcal{H})} \leq k$, then $\size{\cC'} \leq (d-1)!  k^{d-1}$. That is, if a hitting set of the hypergraph $\cH$ is bounded by $k$, then the number of large cores without containing any significant core as a subset is at most $(d-1)!  k^{d-1}$.
\end{lem}
\begin{proof}

Let us consider the set system of all cores in $\cC'$. Note that the number of elements present in each core in $\cC'$ is at most $d-1$. To reach a contradiction assume that 
$\size{\cC'} > (d-1)! \cdot k^{d-1}$.
As $\size{\cC'} > (d-1)! \cdot k^{d-1}$, there exists a $(k+1)$-sunflower $\cS'$, by Proposition~\ref{pre:sunflower}~\footnote{Analogous of Proposition~\ref{pre:sunflower}.} holds even if the hyperedges are of size at most $d$. Let $C_1,\ldots,C_{k+1}$ be the sets present in the sunflower $\cS'$ and let $C_{\cS'}$ be the core of $\cS'$. Observe that $C_{\cS'} \neq \emptyset$, otherwise if $C_{\cS'} =\emptyset$, then $\size{\hs(\cH)} > k$.

To complete the proof of this lemma we need the following observation about $C_{\cS'}$.
 \begin{obs}\label{obs:conflict}
    $C_{\cS'}$ is the pairwise intersection of a family of $k+1$ edges in $\cH$. 
 \end{obs}

 The above observation implies that $S_{\mathcal{H}}(C_{\cS'}) > k$ or equivalently $C_{\cS'}$ is a significant core. Note that each $C_i$ contains $C_{\cS'}$ and $C_{\cS'}$ is a significant core. This contradicts the definition of $\cC'$. We have reached a contradiction, and therefore $\size{\cC'} \leq (d-1)!  k^{d-1}$.
\end{proof}

 \begin{rem}
{    The statement of \cite[Lemma~3.5]{ChitnisCEHMMV16} is same as the combination of our Lemmas~\ref{lem:hs_bound} and~\ref{lem:hs_bound1}. The proof of our Lemma~\ref{lem:hs_bound} is same as that of the corresponding part of the proof of \cite[Lemma~3.5]{ChitnisCEHMMV16}. However, the part corresponding to Lemma~\ref{lem:hs_bound1} in the proof of \cite[Lemma~3.5]{ChitnisCEHMMV16} is incomplete. We give complete proofs of Lemmas~\ref{lem:hs_bound} and~\ref{lem:hs_bound1} in this paper. In particular, inside the proof of \cite[Lemma~3.5]{ChitnisCEHMMV16}, they have made a claim (without a proof) which is equivalent to the statement of Observation~\ref{obs:conflict}.}
 \end{rem}

Now, we prove Observation~\ref{obs:conflict}.

\begin{proof}[Proof of Observation~\ref{obs:conflict}]
 Let $A_i$ be a set of at least $10dk$ edges that form a sunflower with core $C_i$, where $i \in [k+1]$.
 Observe that this is possible as each $C_i$ is a large core.  Before proceeding further, note that $C_i \cap C_j =C_{\cS'}$ and $(C_i \setminus C_{\cS'}) \cap (C_j \setminus C_{\cS'}) =\emptyset$ for all $i , j \in [k+1]$ with $i \neq j$. 
 
  Consider $B_i \subseteq A_i$ such that for each $F \in B_i$, $F \cap C_j = C_{\cS'} ~\forall j \neq i$ and 
  $\size{B_i} \geq 9dk$. First, we argue that $B_i$ exists for each $i \in [k+1]$. Recall that for each $j\in [k+1]$, $\size{C_j} \leq d-1$. Note that any vertex belongs to at most one set $F \setminus C_i$. Also, for any pair of edges $F_1,F_2 \in A_i$, $(F_1 \setminus C_i) \cap (F_2 \setminus C_i)= \emptyset$. Thus, using the fact that $C_i \cap C_j =C_{\cS'}$ for $i \neq j$, a vertex in $C_j \setminus C_{\cS'}$ can belong to at most one edge in $A_i$. This implies that there are at most $(d-1)k < dk$ sets $F$ in $A_i$ such that $F \cap C_j \neq C_{\cS'}$ for some $j\neq i \in [k+1]$. Therefore, the number of edges $F\in A_i$ such that $F \cap C_j = C_{\cS'}~\forall j\neq i \in [k+1]$ is at least $10dk - dk = 9dk$. So, $B_i$ exists as stated.

  Next, we argue that there exists $k+1$ edges
 $F_1,\ldots,F_{k+1}$ such that $F_i \in B_i ~\forall i \in [k+1]$ and $F_i \cap F_j =C_{\cS'}$ for all $i , j \in [k+1]$ with $i \neq j$. We show the existence of the $F_i$'s inductively. For the base case, take any arbitrary edge in $B_1$ as $F_1$.
 Assume that we have chosen $F_1,\ldots, F_p$, where $1\leq p \leq k$, such that the required conditions hold. We will
  show that there exists $F_{p+1} \in B_{p+1}$ such that $F_i \cap F_{p+1}=C_{\cS'}$ for each $ i \in [p]$. By construction of $B_{i}$'s, no edge in $B_{p+1}$ intersects with $C_{i} \setminus C_{\cS'}, i \leq p$; but every edge in $B_{p+1}$ contains $C_{\cS'}$. Also, none of the chosen edges out of $F_1,\ldots,F_{p}$, intersects $C_{p+1} \setminus C_{\cS'}$. So, if we can select an edge  $F \in B_{p+1}$ such that $F \setminus C_{p+1}$ is disjoint from
  $F_i \setminus C_i,~\forall i \in [p]$, then we are done. Note that for two edges $F',F'' \in B_{p+1}$, $F' \setminus C_{p+1}$ and $F'' \setminus C_{p+1}$ are disjoint. Consider the set $B'_{p+1} \subseteq B_{p+1}$ such that each edge $F \in B'_{p+1}$ intersects with at least one out of $\{F_1\setminus C_1,\ldots,F_p \setminus C_p\}$. Observe that $\size{B'_{p+1}} \leq dp \leq dk$, because $(F_i \setminus C_i) \cap (F_j \setminus C_j) = \emptyset,~\forall i\neq j \in [p]$ and $\size{F_i}\leq d, i \in [p]$. As $\size{B_{p+1}} \geq 9dk$, we select any edge in $B_{p+1} \setminus B'_{p+1}$ as $F_{p+1}$.  
 \end{proof}

The following lemma provides insight into the structure of $\hat{\mathcal{H}}$ and thereby is the most important part of proving Lemma~\ref{lem:hs_main}.

\begin{lem}
\label{lem:hs_inter}
Let $\hat{\mathcal{H}}=\mathcal{H}_1 \cup \cdots \cup \mathcal{H}_{\alpha \log k}$. If $\size{\hs(\mathcal{H})} \leq k$, then the following two events hold with high probability.
\begin{itemize}
\item[(a)] $\cF_{s} \subseteq \cF(\hat{\mathcal{H}})$, that is, any hyperedge of the hypergraph $\cH$ that does not contain any large core is a hyperedge in the sampled hypergraph $\hat{\cH}$;
\item[(b)] $S_{\hat{\mathcal{H}}}(C) > k,~\forall C \in \cC'$, that is, every large core in the hypergraph $\cH$ that does not contain any significant core as a subset is a significant core in the sampled hypergraph $\hat{\cH}$.
\end{itemize}
\end{lem}
\begin{proof}
First, consider the two claims stated below.
\begin{cl}
\label{clm:hs-int-1}
Let $ i \in [\alpha \log k]$ and $F \in \cF_{s}$. Then $\pr(F \in \cF(\mathcal{H}_i)) \geq \frac{1}{2}$.
\end{cl}
\begin{cl}
\label{clm:hs-int-2}
Let $ i \in [\alpha \log k]$ and  $C \in \cC$. Then $\pr( S_{\mathcal{H}_i}(C) > k) \geq \frac{1}{2}$.
\end{cl}
Claim~\ref{clm:hs-int-1} says that any herperedge in $\cF_s$ is also a hyperedge in $\cH_i$ with probability at least $1/2$, and Claim~\ref{clm:hs-int-2} says that any large core in $\cC'$ is a significant core in $\cH_i$ with probability at least $1/2$.
Before we give the proofs of Claims~\ref{clm:hs-int-1} and \ref{clm:hs-int-2}, we will first see their implications.

Recall that $\hat{\mathcal{H}}=\mathcal{H}_1 \cup \cdots \cup \mathcal{H}_{\alpha \log k}$. Using Claims~\ref{clm:hs-int-1} and \ref{clm:hs-int-2}, we get the followings for $F \in \cF_s$ and $C \in \cC'$, respectively.
\begin{eqnarray*}
	\pr( F \notin \cF (\hat{\mathcal{H}})) \leq \left(1-\frac{1}{2} \right)^{\alpha \log k} 
	\leq \frac{1}{k^{\alpha}}.
\end{eqnarray*}
and
\begin{eqnarray*}
	\pr( S_{\hat{\mathcal{H}}}(C) \leq k) \leq  \left(1-\frac{1}{2} \right)^{\alpha \log k} 
	\leq \frac{1}{k^{\alpha}}.
\end{eqnarray*}

 Using the union bound together with Lemma~\ref{lem:hs_bound}, we can deduce the following
 \begin{eqnarray*}
 	\pr(\cF_{s} \nsubseteq \cF(\hat{\mathcal{H}})) \leq \sum\limits_{F \in \cF_{s}} 
 	\pr( F \notin \cF(\hat{\mathcal{H}}) ) \leq \frac{d!  (10k)^d}{k^{\alpha}} 
 	\leq \frac{1}{k^{98}}
\end{eqnarray*}
and
\begin{eqnarray*} 
	\pr(\exists ~ C \in \cC' ~\mbox{such that} ~S_{\hat{\mathcal{H}}}(C)  \leq k) 
	\leq \sum\limits_{C \in \cC'} \pr( S_{\hat{\mathcal{H}}}(C) \le k) 
	\leq \frac{(d-1)!k^{d-1}}{k^{\alpha}} \leq \frac{1}{k^{99}}.
\end{eqnarray*}
Note that we have used the fact that $d$ is a constant independent of $k$.
Hence, 
 $$
 \pr(\cF_{s} \nsubseteq \cF(\hat{\mathcal{H}}) ~\mbox{or}~ \exists ~ C \in \cC' ~\mbox{such that} ~S_{\hat{\mathcal{H}}}(C) \leq k) \leq \frac{2}{k^{98}}.
 $$

This implies that with high probability,  $\cF_{s} \subseteq \cF(\hat{\mathcal{H}})$ and $S_{\hat{\mathcal{H}}}(C) > k,~ \forall C \in \cC'$
\end{proof}

We now come back to the proofs of Claims~\ref{clm:hs-int-1} and~\ref{clm:hs-int-2}.
\begin{proof}[Proof of Claim~\ref{clm:hs-int-1}]
Without loss of generality, we will prove the statement for the graph $\mathcal{H}_1$. Let $h:U(\mathcal{H}) \rightarrow [\beta k]$ be the random hash function used in the sampling of $\mathcal{H}_1$. Observe that  by the construction of $\mathcal{H}_1$, $F \in \cF(\mathcal{H}_1)$ if the following two conditions hold.
\begin{itemize}
\item  $ h(u)= h(v)$ if and only if $u=v$ for all $u,v \in F$.
\item For any $F' \neq F$ and $F' \in \cF(\cH)$, $F'$ and $F$ differ in the color of at least one vertex.
\end{itemize}
Hence, $\pr(F \notin \cF(\mathcal{H}_1)) \leq \sum\limits _{u,v \in F:u \neq v}\pr(h(u) = h(v)) + \pr (\cE_1)$,
where $\cE_1$ is the event defined as follows
$$
\mbox{$\cE_1$: $\exists$ an edge $F' \in \cF(\cH)$ such that $ F' \neq F$ and $\{h(z):z \in F\} = \{h(z):z \in F'\}$}.
$$

Before we bound the probability of the occurrence of $\cE_1$, we show the existence of a set $D \subseteq U(\cH) \setminus F$ of bounded cardinality such that each edge in $\cF(\cH) \setminus \{F\}$ intersects with $D$.

\begin{obs}
Let $F \in \cF_s$. There exists a set $D \subseteq U(\cH) \setminus F$ such that  each edge in
$\cF(\cH) \setminus \{F\}$ intersects with $D$ and $\size{D} \leq 2^{d+5}  d^2 k$.
\end{obs}
\begin{proof}
 For each non-empty $C \subset F$, consider the hypergraph $\cH_C$  such that $U(\cH_C)=U(\cH) \setminus C$ and $\cF(\cH_C)$ = $\{F' \setminus C : F' \in \cF(\cH)~\mbox{and}~F' \cap F=C \}$.
First, we prove that the size of $\hs(\cH_C)$ is at most $d S_{\cH}(C)$. For the sake of contradiction,
assume that $\size{\hs(\cH_C)} > d  S_{\cH}(C)$. Then we argue that there exists a maximal packing $\cF' \subseteq \cF(\cH_C)$ such that $\size{\cF'} > S_{\cH}(C)$. If $\size{\cF'} \leq S_{\cH}(C)$, then the vertex set $\{w:w \in F', F' \in \cF'\}$ is a hitting set of $\cH_{C}$ and it has size at most $dS_{\cH}(C)$, which is a contradiction. Therefore, there is a maximal packing $\cF' \subseteq \cF(\cH_C)$ and $\size{\cF'} > S_{\cH}(C)$. Observe that the set of edges $\{F''\cup C : F'' \in \cF' \}$ forms a $t$-sunflower in $\cH$ where $t>S_{\cH}(C)$ and this contradicts the definition of $S_{\cH}(C)$.

The required set $D$ is defined as
$$
	D := \left(\hs(\cH) \setminus F\right)  \cup \left(\bigcup\limits_{C \subset F, C \neq \emptyset} \hs(\cH_C)\right).
$$
$D$ is the desired set because of the followings.
\begin{itemize}
\item If a hyperedge  ${F}^*$ in $\cF(\cH) \setminus \{F\}$ does not intersect with $F$, then it must intersect with $\hs(\cH) \setminus F$;
\item If a hyperedge  ${F}^*$ in $\cF(\cH) \setminus \{F\}$ intersects with $F$, then it must intersect with 
$\hs(\cH_C)$ for some non-empty $C \subset F$. So, each hyperedge in $\cF(\cH) \setminus \{F\}$, intersects with $D$.
\end{itemize}
\color{black}
 Now, we bound the size of $D$. 
Since $\size{\hs(\cH)} \leq k$ and $\size{\hs(\cH_C)}\leq dS_{\cH}(C)$, we have
\begin{eqnarray*}
\size{D} \leq \size{\hs(\cH)} + \size{\bigcup\limits_{C \subset F} \hs(\cH_C)} \leq k + \sum\limits_{C \subset F}dS_{\cH}(C)
\leq k + 2^d\cdot d \cdot 10dk \leq 2^{d+5}d^2k.
\end{eqnarray*}
The last inequality follows from the fact that $F$ does not contain any large core.
\end{proof}

With respect to the set $D$, we define another event $\cE_2 \supseteq \cE_1$ and we bound $\pr(\cE_2)$. Let
$$
	\mbox{$\cE_2$: for some $y \in F$ there exists $ z \in D$ such that $h(z)=h(y)$.}
$$
So, $$\pr(\cE_2) \leq d \frac{\size{D}}{\beta  k } = \frac{d \cdot 2^{d+5}d^2k}{\beta k}=\frac{d^3 2^{d+5}}{\beta}< \frac{1}{10}.$$\\
 
The last inequality holds as $\beta=100d^32^{d+5}$. Putting everything together,
\begin{eqnarray*}
\qquad\qquad\qquad \pr(F \notin \cF(\cH_1) ) \leq \sum\limits _{u,v \in F:u \neq v}\pr(h(u) = h(v)) + \pr (\cE_1)
\leq  \frac{d^2}{\beta k }+\pr (\cE_2)
\leq \frac{d^2}{\beta k }+ \frac{1}{10}
< \frac{1}{2}.
\end{eqnarray*}
\end{proof}

\begin{proof}[Proof of Claim~\ref{clm:hs-int-2}]
Without loss of generality, we will prove the statement for the graph $\mathcal{H}_1$. Let $h:U(\mathcal{H}) \rightarrow [\beta k]$ be the random hash function used in the sampling of $\mathcal{H}_1$. 

Let $\cS$ be the sunflower with core $C$ and $\cF'$ be an arbitrary set of $10dk$ hyperedges corresponding to sunflower $\cS$. 
\begin{obs}\label{obs:des}
With probability $3/4$, there exists a partition of $\cF'$ into equivalence classes 
$T_1,\ldots,T_t$ such that
\begin{itemize}
\item $\bigcup\limits_{x \in F_1 \setminus C} \{ h(x) \} = \bigcup\limits_{x \in F_2 \setminus C} \{ h(x) \}$ if $F_1$ and $F_2$ belong to the same equivalence class, and
\item $\left(\bigcup\limits_{x \in F_1 \setminus C} \{ h(x) \} \right) \cap \left(\bigcup\limits_{x \in F_2 \setminus C} \{ h(x) \} \right)=\emptyset$ if $F_1$ and $F_2$ belong to different equivalence classes.
\item $t \geq 2k$.
\color{black}
\end{itemize} 

\end{obs}
\begin{proof}
For $F \in \cF'$, let $X_F$ be the indicator random variable that takes value $1$ if and only if
$\left(\bigcup\limits_{x \in F \setminus C} \{ h(x) \} \right) \cap \left(\bigcup\limits_{x \in F_1 \setminus C} \{ h(x) \} \right) \neq \emptyset$ for some $F_1 \in \cF' \setminus F$. Observe that $t \geq 10dk -X$, where 
 $$
    X=\sum\limits_{F \in \cF'}X_F. 
 $$
Observe
\begin{eqnarray} 
    \pr(X_F=1) &\leq& \sum\limits_{F_1 \in \cF' } \pr\left( \left(\bigcup\limits_{x \in F \setminus C} \{ h(x) \} \right) \cap \left(\bigcup\limits_{x \in F_1 \setminus C}\{ h(x) \} \right) \neq \emptyset \right)\nonumber\\
    &\leq& 10dk\cdot d^2 \cdot \frac{1}{\beta k} \leq \frac{1}{20}.
\end{eqnarray}

\color{black}
So, $\E[X] \leq \frac{1}{20} \cdot 10dk \leq \frac{dk}{2}$. Now,
\begin{align*}
\pr\left(t \leq 2k\right)
&= \pr\left(X \geq 10dk - 2k\right) &\because t = 10dk-X\\
&\leq \frac{\E[X]}{10dk-2k} &\mbox{from Markov Ineqality}\\
&< \frac{1}{4} &\because\E[X] \leq \frac{dk}{2}
\end{align*}
\end{proof}

Let $h(T)=\bigcup\limits_{x \in F}\left\{h(x)\right\}$, where $F \in T$. Let $\cT$ be the equivalence classes such that the following holds for each $T \in \cT$, we have
$$ 
	h(T) \cap \left(\bigcup\limits_{x \in \hs(\cH) \setminus C } \left\{h(x) \right\}
	\right)  =\emptyset.
$$ 
By the fact that $\size{\hs(\cH)} \leq k$ along with Observation~\ref{obs:des}, $\size{\cT} > k$ holds with probability at least $3/4$. Now, consider the following observation:
\begin{obs}
For each $T \in \cT$, there exists a hyperedge $F$ in $\cH_1$ such that $ \bigcup\limits_{x \in F \setminus C} \left\{ h(x) \right\} =h(T)$ with probability at least $1-\frac{1}{100k}$. 
\end{obs}
\begin{proof}
Consider the set of hyperedges 
$$
\cF'' = \left\{F: \bigcup\limits_{x \in F \setminus C}\left\{ h(x) \right\} = h(T) \right\}. 
$$
Any edge outside $\cF''$ has one vertex $z$ such that $z \in \hs(\cH) \setminus C$ or $ h(z) \in h(T')$ for some 
$T' \in \cT \setminus \{T\}$. By the construction of $\cT$ and by the description of the algorithm, there exists a hyperedge $F$ in $\cH_1$ such that $\left(\bigcup\limits_{x \in F \setminus C} \{ h(x) \} \right)=h(T)$ and the following event $\cE$ holds. $\cE: h(u)=h(v)$ if and only if $u=v$ for all $u,v \in F$. 
$$\pr(\cE^c) \leq \sum\limits_{u,v \in F}\pr(h(u)=h(v))\leq \frac{d^2}{\beta k} \leq \frac{1}{100k}.$$
So, $\pr(\cE) \geq 1- \frac{1}{100k}$.
\end{proof}
From the above observation, there exist at least $\size{\cT}$ hyperedges in $\cH_1$ that form a sunflower with $C$ with probability at least $1-\frac{1}{100k}(k+1) \geq \frac{49}{50}$. As $\pr(\size{\cT} > k) \geq \frac{3}{4}$, $S_{\mathcal{H}_1}(C) > k$ holds with probability $\frac{49}{50}\cdot \frac{3}{4}>\frac{1}{2}$.
%
%
%

\end{proof}
Now, we have all the ingredients to prove Lemma~\ref{lem:hs_main}.

\begin{proof}[Proof of Lemma~\ref{lem:hs_main}]
First, since $\hat{\mathcal{H}}$ is a subgraph of $\mathcal{H}$, a minimum hitting set of $\mathcal{H}$ is also a hitting set of $\hat{\mathcal{H}}$. To complete the proof of this lemma we only need to show that when $\size{\hs(\cH)} \leq k$ then a minimum hitting set of $\hat{\mathcal{H}}$ is also a hitting set of $\mathcal{H}$. By Lemma~\ref{lem:hs_inter}, it is true that with high probability $\cF_{s} \subseteq \cF(\hat{\mathcal{H}})$ and $S_{\hat{\mathcal{H}}}(C) > k$ if $C \in \cC'$. It is enough to show that when $\cF_{s} \subseteq \cF(\hat{\mathcal{H}})$ and $S_{\hat{\mathcal{H}}}(C) > k, ~ \forall C \in \cC{'}$, then a minimum hitting set of $\hat{\mathcal{H}}$ is also a minimum hitting set of $\mathcal{H}$. 

First we show that each significant core intersects with $\hs(\mathcal{H})$. Suppose there exists a significant core $C$ that does not intersect with $\hs(\mathcal{H})$. Let $\cS$ be a $r$-sunflower in $\mathcal{H}$, $r >k$, such that $C$ is the core of $\cS$.  
Then each of the $r$ petals of $\cS$ must intersect
with $\hs(\mathcal{H})$. But the petals of any sunflower are disjoint. 
This implies $\size{\hs(\mathcal{H})} \geq r > k$, which is a contradiction. So, each significant core intersects with $\hs(\mathcal{H})$. As large cores are significant, each large core also intersects with $\hs(\mathcal{H})$.

Let us construct a sub-hypergraph $\tilde{\cH}_{1}$ of $\cH$ with the following definition: Take a large core $C_1$ in $\mathcal{H}$ that contains a significant core $C_2$ as a subset. Let $\cS_1$ be a sunflower with core $C_1$. Let $\cS_2$ be a sunflower with core $C_2$ that has more than $k$ petals. Note that there can be at most one hyperedge $F_1$ of $\cS_1$ that is also present in $\cS_2$. We delete all hyperedges participating in $\cS_1$ except $F_1$. The remaining hyperedges remain the same as in $\mathcal{H}$. Notice that a hitting set of $\tilde{\cH}_{1}$ is also a hitting set of $\cH$; the significant core $C_2$ remains significant in $\tilde{\cH}_{1}$. Thus, any hitting set of $\tilde{\cH}_{1}$ must intersect with $C$ and therefore, must hit all the hyperedges of $\cS_1$. We can think of this as a reduction rule, where the input hypergraph and the output hypergraph have the same sized minimum hitting sets. Let $\tilde{\cH}$ be a hypergraph obtained after applying the above reduction rule exhaustively on $\cH$. The following properties must hold for $\tilde{\cH}$:
\begin{itemize}
    \item[(i)]
        $\hs(\cH) = \hs(\tilde{\cH})$,
        
    \item[(ii)]
        all large cores in $\tilde{\cH}$ do not contain significant cores as subsets, and
    
    \item[(iii)]
        all hyperedges of $\cF_s$ in $\cH$ are still present in $\tilde{\cH}$.
\end{itemize}

By Lemma~\ref{lem:hs_inter}, with high probability we have $S_{\hat{\mathcal{H}}}(C) > k$ when $C$ is a large core of $\tilde{\cH}$ that does not contain any significant core as a subset. Note that the arguments in Lemma~\ref{lem:hs_inter} can also be made for such large cores without significant cores in $\tilde{\cH}$. Thus, we continue the arguments with the assumption that $S_{\hat{\mathcal{H}}}(C) > k$ when $C$ is a large core of $\tilde{\cH}$ that does not contain any significant core as a subset.

Now we show that when $\size{\hs(\cH)}\leq k$, $\hs(\tilde{\cH})=\hs(\hat{\cH})$. We know that $\cF_{s} \subseteq \cF(\tilde{\mathcal{H}})$. That is, any edge that does not contain any large core as a subset, is present in $\tilde{\cH}$. Each hyperedge in $\cF_{s}$ must be covered by any hitting set of $\mathcal{H}$, as well as any hitting set of $\tilde{\mathcal{H}}$ and $\hat{\mathcal{H}}$. Now, it is enough to argue that a hyperedge $F \in \cF(\tilde{\mathcal{H}}) \setminus \cF_s$, must be covered by any hitting set of $\hat{\mathcal{H}}$. Note that each $F \in \cF(\tilde{\mathcal{H}}) \setminus \cF_s$ contains a large core, say $\hat{C}$, which does not contain a significant core as a subset. By our assumption, $\hat{C}$ is a significant core in $\hat{\mathcal{H}}$ and therefore, must be hit by any hitting set of $\hat{\mathcal{H}}$.

Putting  everything together, when $\size{\hs(\mathcal{H}}) \leq k$, each edge in $\mathcal{H}$ is covered by any hitting set of $\hat{\mathcal{H}}$. Thus, $\hs(\mathcal{H})=\hs(\hat{\mathcal{H}})$.
\end{proof}

\section{Algorithms for \dphs{} 
}
\label{sec:hs_nonprom}


We will now prove the following result.
\begin{theo}
\label{theo:hs_np_bis-relaxed}
\dphs{} can be solved with $\tOh_d(k^{2d^2})$ \gpis{} queries with high probability.
\end{theo}
Note that the above result together with the algorithm for \npphs{} that makes $\tOh_d \left( k^d\log n \right)$ \gpis{} 
queries (Theorem~\ref{theo:hs_np_bise}(i)), implies an algorithm for \dphs{} that makes 
$\tOh_d\left( \min\left\{ k^d\log n, k^{2d^2}\right\}\right)$ \gpis{} queries proving the result in Theorem~\ref{theo:hs_np_bise}(ii).


\begin{proof}[Proof of Theorem~\ref{theo:hs_np_bis-relaxed}]
By Observation~\ref{obs:const-prob-boost}, it is enough to give an algorithm that solves \dphs{} with probability at least $2/3$ by using
 $\Oh_{d}\left( k^{2d^2} \right)$ \gpis{} queries. 

We choose a random hash function $h:U(\cH)\rightarrow [\gamma k^{2d}]$, where $\gamma =100  9^d d^2$ {(recall from Section~\ref{sec:prelim} that choosing a said random hash function is about coloring the vertices uniformly and independently at random)}. 
Let $U_{i} = \left\{ u \in U(\cH) : h(u) =i \right\}$,
  where $i \in \left[\gamma k^{2d}\right]$. Note that $U_i$s form a partition of $U(\cH)$, where some 
  of the $U_i$s can be empty. We 
make a \gpis{} query with input $(U_{i_1},\ldots,U_{i_d})$ for each $1 \leq i_1< \ldots < i_d \leq \gamma k^{2d}$ such that  $U_{i_{j}} \neq \emptyset$ for all $j \in [d]$. Recall that the output of a \gpis{} query is Yes or No. We create a hypergraph $\hat{\cH}$ where we create a vertex for each part $U_i$, $i \in [\gamma k^{2d}]$. By abuse of notations, we will denote by
$$
    U(\hat{{\cH}}) =\left\{U_1, \ldots,U_{\gamma k^{2d}}\right\}
$$ 
and 
$$
    \cF(\hat{\cH}) = \left\{(U_{i_1},\ldots,U_{i_d}):~\mbox{\gpis{} answers ``yes'' to the input $(U_{i_1},\ldots,U_{i_d})$}\right\}.
$$
Observe that we make $\Oh_{d}\left( k^{2d^2} \right)$ queries 
to the \gpis{} oracle. We find $\hs(\hat{\cH})$ and report $\size{\hs(\cH)} \leq k$ if and only if $\size{\hs(\hat{\cH})} \leq k$. 

For the hitting set $HS({\cH})$, consider the set $S' = \left\{ U_i~\vert~ \exists u\in \hs(\cH), h(u)=i \right\}$. 
Then $S'$ is a hitting set for $\hat{\cH}$. So, $\size{\hs(\hat{\cH})}\leq \size{\hs(\cH)}$, and  
if $\size{\hs(\cH)} \leq k$, then $\size{\hs(\hat{\cH}
)} \leq k$.
Now, the correctness of our query procedure follows directly from the following claim.
\begin{cl}
\label{claim:hs_rep}
 If $\size{\hs(\hat{\cH})} \leq k$, then $\size{\hs(\cH)} \leq k$ with probability at least $2/3$. 
\end{cl}
\noindent
The remaining part of the proof will prove the above claim.

Let $\cR$ be a fixed $k$-representative set corresponding to $\cH$ obtained from Proposition~\ref{pre:rep-set-system}  and let $\mathcal{H}'$ be a set system obtained from $\cR$ as described in Corollary~\ref{coro:rep-set-system}. Consider the set $U(\mathcal{H}')$. Note that $\size{\mathcal{F}(\mathcal{H}')} \leq {{k+d}\choose{d}}$ and 
$\size{U(\mathcal{H}')} \leq d \cdot {{k+d} \choose d}$. Let $\cE_1$ be the event that all the vertices in $U(\mathcal{H}')$
are uniquely colored, i.e., $\cE_1$: $h(u)=h(v)$ if and only if $ u=v$, where $u,v \in U(\mathcal{H}')$.

Now we lower bound the probability of the event $\cE_1$. As usual, let $\cE_1^c$ denote the complement of the event $\cE_1$. Therefore,
$$\pr(\cE_1^c) \leq \sum\limits_{u,v \in U(\mathcal{H}')} \pr(h(u) = h(v)) 
\leq \sum\limits_{u,v \in U(\mathcal{H}')} \frac{1}{\gamma k^{2d}} \leq \frac{\size{U(\mathcal{H}')}^2}{\gamma k^{2d}} < \frac{1}{3}.$$

So, $\pr(\cE_1)  \geq  \frac{2}{3}$. Let ${\sf Prop}$ be the property that for each $F \in \mathcal{F}(\mathcal{H}')$, there is an ``equivalent'' hyperedge in $\cF(\hat{\cH})$. More specifically, ${\sf Prop}$ is the following property: For each $(u_1,\ldots,u_d) \in \mathcal{F}(\mathcal{H}')$, the hyperedge $(U_{h(u_1)},\ldots,U_{h(u_d)})$ belongs to $\cF(\hat{\cH})$ for all $i \in [d]$.

From the definition of the \gpis{} query oracle, observe that the property ${\sf Prop}$ is true whenever the event $\cE_1$ occurs. If we show that the occurrence of ${\sf Prop}$ implies that $\size{\hs(\cH)} \leq k$ if and only if $\size{\hs(\hat{\cH})} \leq k$, we are done.

For the rest of the proof, assume that ${\sf Prop}$ holds. Let us define a function $f:U(\hat{\cH}) \rightarrow U({\mathcal{H}'}) \cup \{\psi\}$ as follows. For each $i \in [\gamma  k^{2d}]$, if $ h(u) =i$ and $u \in U(\mathcal{H}')$, then $f(U_i)=u$. Otherwise, $f(U_i)=\psi$.

Let $\size{\hs(\hat{\cH})} =k' \leq k$. Let $\hs(\hat{\cH})=\{X_1,\ldots,X_{k'}\}\subseteq U(\hat{\mathcal{H}})$. Consider the vertex set $U'=\{f(X_i):i \in [k'],f(X_i) \neq \psi \} \subseteq U(\mathcal{H}')$ which is of size at most $k$. As $\hs(\hat{\cH})$ is a hitting set of $\hat{\cH}$, $U'$ covers all the hyperedges present in $\mathcal{F}(\mathcal{H}')$. Hence by Corollary~\ref{coro:rep-set-system}, $\size{\hs(\cH)} \leq k$.
\end{proof}

\section{Lower bound for \dphs}
\label{sec:lowerbound-restate}
\noindent
We will prove the following result in this Section.

\begin{theo}[Restatement of Theorem~\ref{theo:lowerbound}]
\label{theo:lowerbound-restate}
Let $n,k,d \in \N$ with $d\leq k$ and $n \geq k+d$. Any algorithm, with \gpis{} query access to any hypergraph $\cH$ having $n$ vertices, that decides whether $\hs({\cH})
\leq k$ or $\hs({\cH}) \geq k+1$ with 
probability $2/3$, makes at least $\Omega\left(\binom{k+d}{d}\right)$ queries. 
\end{theo}

\begin{rem}
    The proof of Thoerem~\ref{theo:lowerbound-restate} can be directly adapted for \gpise{} query, i.e., we can show
    that any algorithm with \gpise{} query access to any hypergraph $\cH$ that decides whether $\hs({\cH}) \leq k$ or $\hs({\cH}) \geq k+1$, with 
    probability $2/3$, makes at least $\Omega\left(\binom{k+d}{d}\right)$ queries.
\end{rem}

We use the framework by Eden and Rosenbaum~\cite{DBLP:conf/approx/EdenR18} to prove the above theorem via a reduction from $\mbox{{\sc Disjointness}}_{N}
$ problem in the Yao's two party communication model. 
Given two vectors ${\bf x}$ and ${\bf y}$ in $\{0,1\}^{N}$, we say ${\bf x}$ and ${\bf y}$
{\em intersect} if there exists $i \in [N]$ such that $x_{i} = y_{i} = 1$.\footnote{For a vector ${\bf z}\in \{0,1\}^{N}$, 
$z_{i}$ denotes the $i$-th coordinate of the vector ${\bf z}$.} Otherwise, we say 
${\bf x}$ and ${\bf y}$ are {\em disjoint}.
In the $
\mbox{{\sc Disjointness}}_N$ problem, we have two players Alice and Bob, where Alice has a vector ${\bf x} \in \{0,1\}^{N}$ and Bob has a vector ${\bf y} \in \{0,1\}^{N}$. Note that Alice does not know about Bob's vector and Bob does not know about Alice's vector. The goal of the $\mbox{{\sc Disjointness}}_N$ problem is for Alice and Bob to communicate bits between each other following a pre-decided protocol in order to decide if ${\bf x}$ and ${\bf y}$ intersect or not. 
The {\em communication complexity} of $\mbox{{\sc 
Disjointness}}_N$ is defined as the minimum number of bits communicated 
between Alice and Bob, by the best protocol in the worst case, to solve 
$\mbox{{\sc Disjointness}}_N$ with probability at least $2/3$~\cite{DBLP:books/daglib/0011756}. It is well known that the communication complexity of $\mbox{{\sc Disjointness}}_N$ is $\Omega(N)
$~\cite{DBLP:books/daglib/0011756}. The lower bound holds even if it is known from beforehand that either ${\bf x}$ and ${\bf y}$ are disjoint, or there exists exactly one $i\in [N]$ such that $x_{i}=y_{i}=1$,
see~\cite{DBLP:books/daglib/0011756}.

\begin{proof}[Proof of Theorem~\ref{theo:lowerbound-restate}]
Let ${\bf x} \in \{0,1\}^{N}$ and ${\bf y} \in \{0,1\}^{N}$, where $N = \binom{k+d}{d}$
be the inputs of Alice and Bob, respectively. Moreover, assume that either ${\bf x}$ and ${\bf y}$ 
are disjoint or there exists exactly one $i\in [N]$ 
such that $x_i=y_i=1$.  Fix a bijection $\phi:N\rightarrow 
\Sigma_{d}$, where $\Sigma_{d}$ denote the collection of all $d$-sized subsets of $[k+d]$. Let $\cH({\bf x},{\bf y})$ 
be the hypergraph (with $
[n]$ as the vertex set), that can be uniquely determined from ${\bf x}$ 
and ${\bf y}$, having hyperedges according to the following rule: for each $i\in 
[N]$, the particular combination of $\binom{k+d}{d}$ indicated by $\phi(i)$ is a hyperedge in $\cH({\bf x},{\bf y})$ if at least one of $x_i$ and $y_i$ is $0$. 
Note that no hyperedge in $\cH({\bf x},{\bf y})$
contains a vertex from $[n]\setminus [k+d]$. Moreover, $\cH({\bf x},{\bf 
y})$ satisfies the following observation because of the particular nature of ${\bf x}$ and ${\bf y}$.
\begin{obs}\label{obs:hard-inst-graph}
\begin{itemize}
\item[(i)] 
  There exists at most one $d$-sized subset of $[k+d]$ that is not a hyperedge in $\cH({\bf x},{\bf y})$.
\item[(ii)] 
  If ${\bf x}$ and ${\bf y}$ are disjoint then each $d$-sized subset of $[k+d]$ is a hyperedge in $\cH({\bf x},{\bf y})$, and
  therefore the minimum size of any hitting set of $\cH({\bf x},{\bf y})$ is $k+1$. Otherwise, there is exactly one $d$-sized subset 
  of $[k+d]$ that is not a hyperedge in $\cH({\bf x},{\bf y})$, and therefore $\cH({\bf x},{\bf y})$ has a hitting set of size $k$.
\end{itemize}
\end{obs}
The above observation follows from the construction of $\cH({\bf x},{\bf y})$ along with the fact that either 
${\bf x}$ and ${\bf y}$ are disjoint or there exists exactly one $i\in [N]$ such that $x_i=y_i=1$.
\remove{\begin{proof}
If ${\bf x}$ and ${\bf y}$ are disjoint, then either $x_i=0$ or $y_i=0$ for each $i\in [N]$. Then by description of $\cH({\bf x},{\bf y})$, each $d$-sized subsets in $S_d$ forms a hyperedge in $\cH({\bf x},{\bf y})$. As no hyperedge contains any vertex from $[n]\setminus [k]$, we can assume that any (minimal) hitting set of $\cH({\bf x},{\bf y})$ is
\end{proof}}

To reach a contradiction, assume that there exists an algorithm {\sc Alg} that makes $o\left(\binom{k+d}{d}\right)$ \gpis{} queries to $\cH({\bf x},{\bf y})$ and decides whether $\hs(\cH({\bf x},{\bf y}))\leq k$ or $\hs(\cH({\bf x},{\bf y}))= k+1$. 
Now we give a protocol for $\mbox{{\sc Disjointness}}_N$ with $o\left( \binom{k+d}{d} \right)=o\left( N \right)$ bits of communication. Alice and Bob run {\sc Alg} on $\cH({\bf x},{\bf y})$. Let {\sc Alg} asks for a \gpis{} query with input 
$A_1,\ldots,A_d$. Note that $A_1,\ldots,A_d$ are non-empty and pairwise disjoint. Without loss of generality, 
we can assume that $A_1,\ldots,A_d \subset [k+d]$ as no hyperedge in $\cH({\bf x},{\bf y})$ contains any vertex from 
$[n]\setminus [k+d]$. Now, we describe how Alice and Bob simulate each \gpis{} query by communicating at most $2$ bits.

\begin{description}

\item {\bf At least one $A_i$ has at least two vertices from $[k+d]$:} 
  By Observation~\ref{obs:hard-inst-graph} (i), in this case, there exists a 
  hyperedge having a vertex in each $A_i$. So, Alice and Bob can answer to 
  any such \gpis{} query without any communication. 

\item {\bf Each $A_i$ is a set of singleton vertex from $[k+d]$:} 
  In this case, Alice and Bob need to determine whether the vertices in $A=\bigcup\limits_{i=1}^d A_i \subseteq [k+d]$ 
  form a hyperedge in $\cH({\bf x},{\bf y})$. Let $j=\phi^{-1}\left(A\right)$. From the description of $\cH({\bf x},{\bf y})$, $A$ 
  is a hyperedge if and only if at least one of $x_j$ and $y_j$ is $0$. So, Alice and Bob can know the answer to any such \gpis{} query by communicating their bits at $j$-th index, which is 2 bits of communication.

\end{description}

Hence, Alice and Bob can simulate algorithm {\sc Alg} by using $o(N)$ bits of communication. After simulating {\sc Alg}, 
Alice and Bob reports ${\bf x}$ and ${\bf y}$ intersect if {\sc Alg} reports that $\hs(\cH({\bf x},{\bf y})) \leq k$. Otherwise, 
if {\sc Alg} reports that $\hs(\cH({\bf x},{\bf y})) =k+1$, Alice and Bob report ${\bf x}$ and ${\bf y}$ are disjoint. The 
correctness of the protocol for $\mbox{{\sc 
Disjointness}}_N$ follows from the existence of algorithm {\sc Alg} and Observation~\ref{obs:hard-inst-graph} (ii).
\end{proof}
\color{black}

\section{Conclusion}
\label{sec:cut}
\noindent
In this paper, we proved that the query complexities of \dphs{} and \npphs{} problems, using GPIS query, to be $\tOh_{d}\left( \min\left\{ k^{d} \log n, k^{2d^{2}} \right\}\right)$ and $\tOh_{d} \left( k^{d} \log n\right)$, respectively. We have also considered \gpise{} query oracle, which is basically a \gpis{} query oracle that also provides existence of an edge with an arbitrary example, and therefore is a stronger query oracle than \gpis{}. We showed that both \dphs{} and \npphs{} can be solved by using $\tOh_d(k^d)$ \gpise{} queries. To complement our upper bounds,
we proved an almost matching lower bound of $\Omega\left( \binom{k+d}{d}\right)$ \gpise{}\footnote{The lower bound is proved for \gpis{} query oracle but can be directly adapted for \gpise{} query oracle.} queries for both of these problems.

We think that the $\log n$ term in the query complexity of \dphs{} is not required, and therefore we believe that the query complexity of \dphs{} using GPIS query should be $\widetilde{\Theta}_{d} \left( k^{d}\right)$.
Unlike the \dphs{} problem, we believe that the query complexity of \npphs{} using \gpis{} query should be $\widetilde{\Theta}_{d} \left( k^{d} \log n \right)$. If this is indeed true then there is a {\emph separation} between \gpis{} and \gpise{} query oracles, as we have already showed that the query complexity of \npphs{} using \gpise{} query is $\widetilde{\Theta}_d\left(k^d\right)$. 

Ron and Tsur~\cite{RonT16} studied the power of an example in the context of classical subset/group queries. They showed separation results between the classical subset queries and subset queries with example. If our hunch about \npphs{} is true with respect to GPIS query, then this would be another problem showing the power of an example in the context of GPIS query, a particular kind of subset query.

\bibliographystyle{alpha}
\bibliography{reference}


\end{document}